\documentclass[12pt]{amsart}
\usepackage{latexsym,fancyhdr,amssymb,color,amsmath,amsthm,graphicx,listings,comment}
\usepackage[section]{placeins}
\newtheorem{thm}{Theorem}  
\let\paragraph\subsection

\title{Division algebra valued energized simplicial complexes}
\author{Oliver Knill}
\date{August 23, 2020]}
\address{Department of Mathematics \\ Harvard University \\ Cambridge, MA, 02138 }


\begin{document}
\maketitle

\begin{abstract}
We look at matrices $L,g$ defined by a function $h:G \to K$, where
$G$ is a finite set of sets and $K$ is a normed division ring which does not need
to be commutative, nor associative but has a conjugation leading to the
norm $|h|^2 = h^* h$. The target space $K$ can be a normed real division algebra
like the quaternions or an algebraic number field like a quadratic field. 
For parts of the results we can even include Banach algebras like an operator 
algebra on a Hilbert space. The wave $h$ on $G$ then defines connection matrices 
$L,g$ in which the entries are in $\mathbb{K}$. 
We show that the Dieudonn\'e determinants of $L$ and $g$ are both
equal to the Abelianization of the product of all the field values on $G$.
If $G$ is a simplicial complex and $h$ takes values in the units $\mathbb{U}$ of 
$\mathbb{K}$, then $g^*$ is the inverse of $L$ and the sum of the energy values is 
equal to the sum of all the Green function matrix entries $g(x,y)$.
If $K$ is the field $\mathbb{C}$ of complex numbers, we can study the spectrum of 
$L(G,h)$ in dependence of the field $h$. The set of matrices with simple spectrum defines 
a $|G|$-dimensional non-compact K\"ahler manifold that is disconnected in general and 
for which we can compute the fundamental group in each connected component.
\end{abstract}

\section{In a nutshell}

\paragraph{}
Assume $G$ is a {\bf finite set of sets} and $\mathbb{K}$ is a {\bf normed division ring} with 
$1$ and conjugation $*$. The ring $\mathbb{K}$ is not need to be associative nor commutative, it 
can be one of the division algebras $\mathbb{R},\mathbb{C},\mathbb{H},\mathbb{O}$,
or also a quadratic field. Let $\mathbb{U}=\{ |h^* h|=1 \}$ is the set of units in $\mathbb{K}$.
The division algebra relation $|h|^2 = h^* h$ for $*$ division algebras can be 
weakened to a {\bf Banach algebra condition} $|xy| \leq |x| |y|$, in particular if the fields
take values in the units $\mathbb{U} \subset \mathbb{K}$ of elements with norm $1$. 

\paragraph{}
A function $h: G \to \mathbb{K}$ defines the $\mathbb{K}$-value $H(A)=\sum_{x \in A} h(x)$ 
for a subset $A$ of $G$. In the topological case, $h(x)=\omega(x) = (-1)^{{\dim}(x)}$, 
where $H(A)=\chi(A)$ is the {\bf Euler characteristic} of $A$. 
In the case $h(x)=1$ then $H(A)=|A|$ is the {\bf cardinality} of $A$. Define the matrix
$L(x,y)=H(W^-(x) \cap W^-(y))$, where $W^-(x) = \{ y \in G, y \subset x \}$ is the 
{\bf core} of $x$ and $g(x,y) = \omega(x) \omega(y) H(W^+(x) \cap W^+(y))$ with the star 
$W^+(x) = \{ y \in G, x \subset y \}$ of $x$. 

\paragraph{}
Here are summaries of the results. The first part generalizes  \cite{KnillEnergy2020}. The
second part on the spectrum bridges to complex differential geometry. \\
(A) The Dieudonn\'e determinant ${\rm det}(L) = {\rm det}(g) =\prod_x \overline{h}(x)$ 
    in the Abelianization of $\mathbb{K}$. The Study determinant of $L$ and $g$ is 
    ${\rm det}(L) = {\rm det}(g) = \prod_{x \in G} |h(x)|$. (In the Banach algebra case, this would
    be $\leq \prod_x |h(x)|$. \\
(B) If $G$ is a simplicial complex and $h$ is $\mathbb{U}$ valued, 
    then $g^*$ is the inverse of $L$ in the sense
    $g^* L = g L^* = 1$ in $M(n,\mathbb{K})$. This works also in {\bf $C^*$-algebra} settings. \\
(C) If $G$ is a simplicial complex, the energy relation $\sum_{x,y} g(x,y) = H(G)$ holds. \\
(D) For $\mathbb{K}=\mathbb{C}$ and $x \in G$, the deformation 
    $h_t(x)=e^{i t} h(x),h_t(y)=y, y \neq x$ defines a spectral deformation $\lambda_k(t)$, 
    and $\lambda_k(2\pi) \neq \lambda_k(0)$ in general. 
    The rotations for which $\lambda_k(t) \neq \lambda_l(t)$ for $k \neq l$
    define relations for a {\bf finitely presented permutation group} $\Pi(G,h_0)$. \\
(E) For fixed $G$, the $n$-dimensional non-compact {\bf K\"ahler manifold} $M$ of 
    matrices $L(G,h) \in GL(n,\mathbb{C})$ with simple spectrum is in general not connected. 
    The fundamental group $\pi_1(M,h_0)$ of a component containing $L(G,h_0)$ 
    is explicitly computable as $\pi(M,h_0)$ which is the finitely presented group 
    $\Pi(G,h_0)$ in which cyclic relations are omitted. 
    This allows us to construct many K\"ahler manifolds with
    explicitly known fundamental group.

\section{Result (A): Determinants}

\paragraph{}
Assume $G$ is an arbitrary {\bf finite set of sets}. 
Let $\mathbb{K}$ is a {\bf normed division ring} with $1$ and conjugation $*$ and 
norm $|h|^2=h^* h = h h^*$. The ring $\mathbb{K}$ does not need to be associative nor 
commutative, nor does it have to be defined over the reals. 
It can be one of the four real normed division algebras 
$\mathbb{R},\mathbb{C},\mathbb{H},\mathbb{O}$, or a quadratic field like the 
Gaussian rationals $\mathbb{Q}[i]$ which is an example of an {\bf algebraic number field}.

\paragraph{}
A function $h: G \to \mathbb{K}$ defines the $\mathbb{K}$-value $H(A)=\sum_{x \in A} h(x)$
for a subset $A$ of $G$. When considered on simplicial complexes $H$ and $h$ is the same on sets of the
same cardinality is also known as a {\bf valuation} as $H(A \cap B) + H(A \cup B) = H(A) + H(B)$. 
In the topological case, $h(x)=\omega(x) = (-1)^{{\dim}(x)}$, the valuation $H(A)=\chi(A)$ defines the
Euler characteristic of $A$. In the case of constant $h(x)=1$ then $H(A)=|A|$ is the cardinality. 
Define the matrix 
$L(x,y)=H(W^-(x) \cap W^-(y))$, where $W^-(x) = \{ y \in G, y \subset x \}$ is the {\bf core} of $x$ and
$g(x,y) = \omega(x) \omega(y) H(W^+(x) \cap W^+(y))$ with the star 
$W^+(x) = \{ y \in G, x \subset y \}$ of $x$. 

\paragraph{}
The {\bf Leibniz determinant} 
${\rm det}(L) = \sum_{\sigma} {\rm sign}(\sigma) L_{1,\sigma(1)} \cdots L_{n \sigma(n)}$ 
is defined for matrices $L \in M(n,\mathbb{K})$ for any ring $\mathbb{K}$ but it fails the Cauchy-Binet 
relation ${\rm det}(AB)={\rm det}(A) {\rm det}(B)$ in general. Other determinants have therefore
been defined. The {\bf Dieudonn\'e determinant} \cite{Dieudonne1943}
and {\bf Study determinant} \cite{Study1920} both do satisfy the product relation. 
Their definition uses row reduction of $L$ to an upper triangular matrix for its definition.
In order to row reduce, we need the ability to divide, hence the assumption of having a division ring. 
The Dieudonn\'e determinant takes values in the {\bf Abelianization} $\overline{\mathbb{K}}$ 
of $\mathbb{K}$ while the Study determinant takes real values and involves the norms of the product
of diagonal elements after reduction. More about the linear algebra is included in 
the Appendix. The following formula has originally first been 
considered if $h(x) = \omega(x)=(-1)^{{\rm dim}(x)}$ 
in which case, the formula shows that the matrices $L,g$ are 
unimodular integer matrices \cite{Unimodularity}.

\begin{thm}[Determinant formula]
The Dieudonn\'e determinant satisfies 
${\rm det}(L) = {\rm det}(g) = \overline{h}(x_1) \cdot \overline{h}(x_2) \cdots \overline{h}(x_n)$.
\label{1}
\end{thm}

\begin{proof}
The assumption of having $G$ an arbitrary set of sets rather than a simplicial complex
and also not insisting on any ordering has the advantage that we have now a {\bf duality}
between the matrices $L^-(x,y)=H(W^-(x) \cap W^-(y))$ and 
$L^+(x,y) = H(W^+(x) \cap W^+(y))$, where $G$ is replaced by $G^*$, the set of 
sets with the {\bf complements} $x^* = (\bigcup_{x \in G} G) \setminus x$ and 
assigning to the complement set the same value $h(x)$. The determinant
of $g(x,y) = \omega(x) \omega(y) L^+(x,y)$ is the same than the determinant of
$L^+(x,y)$ because multiplying one row or column with $-1$ changes the sign of
the Leibniz determinant. For the Dieudonn\'e determinant, switching rows does not
change anything if $\overline{-1}$ is a commutator like in the case of 
quaternions or octonions, where $i j i^{-1} j^{-1} = -1$. 
To prove the statement, we use induction with respect to the number of elements,
noting that the induction assumption also shows that the determinant is multi-linear in 
each of the terms $h(x_k)$.
Assume therefore we have proven the statement for any $G$ with $|G|=n-1$ or less, we add a new set
$x$ to $G$ and give it a new value $h(x)=X$. Now write down the matrix $L$. Every matrix
entry $L(y,z)$ with $y \subset x, z \subset x$ contains a linear term $a_{yz} + X$,
where $a_{yz}$ is an other element in $\mathbb{K}$. 
All other matrix entries do not depend on $X$. Laplace expansion allows to write
the determinant as a sum of minors and since each minor by induction assumption is 
linear in $X$, ${\rm det}(L)$ is a linear function of $X$. Each $k \times k$-minor not containing
the last row or column is zero. The reason is that the induction assumption works and 
that this minor is linear in $X$ with factor given by the minor when $X=0$. 
Having all minors containing the last row or last column to be zero and the others linear in $X$
shows ${\rm det}(L)= h(x_1)= \dots h(x_{n-1}) X$.
\end{proof}

\paragraph{}
Let us add some initial remarks. More discussion about the motivation to look at fields 
or skew fields $h$ on geometries $G$ is in a discussion section at the end. 
First of all, every minor can be re-interpreted as a 
determinant of a sub structure of $G$ with an adapted energy. This is an additional advantage of 
working with general multi-graphs and not only with simplicial complexes. 
A arbitrary subset $A$ of $G$ is in the same category of objects and then defines
a minor. The proof of the above statement is a bit easier if
$x$ is a maximal element in $G$. The matrix entries $L(y,z)$ for any other set does not involve 
the value of $x$ and the entry $X=h(x)$ only appears in the last row and column of $L$ as linear 
terms $X$ and not $a+X$ as in general. The Laplace expansion then shows that the minor
without last row and column is the slope factor of the determinant which is linear in $X$. 

\paragraph{}
The Abelianization works also for simpler algebraic structures like
monoids and magmas but one usually assumes associativity. Already for
quaternions or octonions, we need the Study or Dieudonne determinant, 
which agree there because $\overline{h}=|h|$. 
In the non-associative case like Jordan algebras or normed Lie algebras,
we have to specify brackets like bracketing from the right 
$abc = (a(bc))$. For $C^*$ algebra or normed Lie algebras one assumes only 
inequalities $|xy| \leq |x| |y|$ and the determinant formula becomes then 
an inequality too in the Study determinant case. 
The matrices $L$ and $g$ are then determined, if an order on $G$ is given. 
For the above result we do not need to have the elements of $G$ ordered in a specific way.
Because the Leibniz determinant does not satisfy Cauchy-Binet 
and also dependents of the ordering of $G$ are reasons that this determinant is not used much. 

\paragraph{}
For $\mathbb{K}=\mathbb{C}$, where we know that there are exactly $n$ complex eigenvalues $\lambda_k$
of $L$, we have $\prod_{k=1}^{|G|} \lambda_k = \prod_{x \in G} h(x)$. 
One can then also use the Leibniz determinant. The Dieudonn\'e determinant is then the same
than the Leibniz determinant but the Study determinant is in the complex case given as
${\rm det}(L) = \prod_{x \in G} |h(x)|$. So, the three determinant definitions are all different. 
The Dieudonn\'e determinant contains in general more information than the Study determinant. 
But there is also an advantage for the Study determinant: one has not to worry about commutators 
but directly can just look at the norm. 

\paragraph{}
Theorem~(\ref{1}) was proven in \cite{Unimodularity} for $h(x)=\omega(x)$ 
by building $G$ as a CW-complex. See \cite{KnillEnergy2020}. In \cite{MukherjeeBera2018}
a proof was given which avoids discrete CW-complexes. The CW proof still works $h$ can 
become more general \cite{EnergizedSimplicialComplexes}.
In each case, ${\rm det}(L)$ is multiplied by $h(x)$ each time a new simplex is added.
We realized in \cite{CountingMatrix} that the unimodularity theorem works even for
arbitrary finite sets of non-empty sets. This is a structure which is also called a {\bf multi-graph}.
This allows then to use {\bf duality} $x \to \hat{x}=(\bigcup_{x} x) \setminus x$ to switch the
{\bf stars} $W^+(x)$ and {\bf cores} $W^-(x)$  which can be seen as unstable and stable manifolds.
If $G$ is a simplicial complex, only $W^-(x)$ is a simplicial complex in general, $W^+(x)$ not. 

\section{Result (B): unit valued fields}

\paragraph{}
In the following result, we assume that $G$ is a simplicial complex, a finite set of non-empty sets
closed under the operation of taking non-empty subsets.
When writing matrices down, we often assume $G$ to be ordered, so that the first row and
first column corresponds to the first element of $G$. When working numerically, we usually make this
assumption by ordering according to size so that the matrices $g^* L$ 
are then in general upper triangular for any simplicial complex $G$. We do not 
have to assume this ordering however. Still, it is important to note that in the non-commutative
case, the order of $G$ matters in the sense that it determines the conjugacy class of the matrix.
If the fields take values in $\mathbb{U}$, we still have a wonderful relation between $g$ and $L$
but we need $G$ to be a simplicial complex. 

\begin{thm}[Green star identity]
Assume $G$ is a simplicial complex.
If $h: G \to \mathbb{U}$ takes values in the units $\mathbb{U}$ of $\mathbb{K}$,
then $g^* L = L g^* = 1$.
\label{2}
\end{thm}
\begin{proof}
Given $x \in G$. Write $(g^* L)(x,x) = \sum_{y} g^*(x,y) L(y,x)$. 
Now $g^*(x,y) L(y,x) = \sum_{x \subset z \subset x} \omega(z)^2 h^*(z) h(z) = |h(x)|^2$.  \\
If $x,y \in G$  and not either $x \subset y$ or $y \subset x$, then there
is no $z$ which contains $y$ and is contained in $x$ so that $(g^* L)(x,y)=0$. 
If $x \subset y$, then, using $h^*(z) h(z) = |h(z)|^2=1$
$$ (g^* L)(x,y) = \sum_{x \subset z \subset y} \omega(x) \omega(z) h^*(z) h(z)
                = \omega(x) \sum_{x \subset z \subset y}  \omega(z) 1  \; .  $$
Now, there are an equal number of elements $z$ between $x$ and $y$ which have $\omega(z)=1$
than elements which have $\omega(z)=0$. This is a general fact as we can can just 
look at the simplex $y$ and remove all elements in $x$, then have a simplex in which there
are equal number of odd and even dimensional simplices (including the empty element). 
This was the place, where we needed that $G$ is a simplicial complex. 
If $y \subset x$, then there is no $z$ with $x \subset z \subset y$ and the sum is zero. 
\end{proof}

\paragraph{}
The etymology for the term ``Green-Star" is as follows:
we can look at $g(x,y)$ as {\bf Green function} entries which depend on the {\bf stars}
$W^+(x)$ and $W^+(y)$, so that we called it the {\bf Green-Star identity}.
The terminology of Green functions is extremely important in mathematical physics: whenever
we have a Laplacian $L$. The term ``star" is an official term in algebraic combinatorics. 

\paragraph{}
The word combination ``Green-Star" is also a bit of a pun because we had been 
blind for a long time. This is documented in blog entries of our quantum calculus blog. 
(Green Star also stands for {\bf Glaucoma}). We needed a many months of attempts and an 
insane amount of experiments to get the formula because we had been looking for expressions 
in which $g(x,y)$ is the Euler characteristic of a sub-complex. To experiment, we correlated 
the entries $g(x,y)$ with the Euler characteristic of various sub-complexes of $G$. 
The solution was not to insist on having simplicial complexes any more and indeed, 
stars $W^+(x)$ are examples of sub-structures of simplicial complexes which in general 
are {\bf not} simplicial complexes. 

\paragraph{}
Also the Green star formula for the matrix entries of the inverse $g$ of $L$
generalizes. While the entries $L(x,y)$ involve the cores of $x$ and $y$, the
entries $g(x,y)$ involve the stars of $x$ and $y$. The formula had been first
developed in the topological case. Remarkable in the constant case $h(x)=1$ is that
$g=L^{-1}$ is isospectral to $L$, which is a {\bf symplectic relation} and that $L,g$ are then 
both {\bf positive definite integer quadratic forms} which are isospectral.
This led to {\bf ispectral multi-graphs}
and a functional equation for the {\bf spectral zeta function} $\zeta(s) = \sum_k \lambda_k^{-s}$
defined by eigenvalues $\lambda_k$ of $L$. Unlike for the Riemann-Zeta function which is the
spectral zeta function of the circle $\mathbb{T}$, and more generally for spectral zeta functions of 
manifolds, we do not have to discard any zero eigenvalue
for connection Laplacians because they are invertible. 

\paragraph{}
It was important in the previous theorem that $h$ is $\mathbb{U}$-valued. The proof shows that
the condition is not only sufficient but also necessary:
the non-diagonal entries are of sums of the form $|h(y)|^2-|h(z)|^2$. The diagonal
entries of $g^* L$ are just $|h(x)|^2$, which shows that $g^* L = 1$ implies
that $L$ is $\mathbb{U}$-valued.
On the other hand, already in the complex case, the matrices 
$L$ and $g$ are symmetric but no more self-adjoint. The spectra are in the complex plane. 
We will see that this has also advantages as we can define K\"ahler manifolds $M=M(G,h)$ 
for any field $h: G \to \mathbb{K}$, where $G$ is a finite set of sets. 

\paragraph{}
The definitions of $L$ and $g$ do not tap into the multiplicative structure of the
algebra but once we multiply, it matters. However, if the simplicial complex $G$
is ordered so that the dimension increases, then $g^* L$ is upper triangular
and $L g^*$ is lower triangular. In the diagonal, we have then terms $|h(y)|^2$
and in the upper or lower part we have sums of expressions which are sums of 
$|h(y)|^2 - |h(z)|^2$ for different pairs of $y,z$. 

\paragraph{}
For $\mathbb{K}=\mathbb{R}$, we had a spectral relation telling that the number of negative
values of $h$ is equal to the number of negative eigenvalues of $L$. This could be rephrased
in that one can ``hear the Euler characteristic" of $G$ \cite{HearingEulerCharacteristic}. 
We do not know how to hear $H(G)$ in general yet. Yes, it is the sum of the matrix entries of
$g$ but we would like to have a formula which gives $H(G)$ in terms of the eigenvalues of $L$. 
Already for $\mathbb{K}=\mathbb{C}$, the spectrum of $L$ and $g$
are in the complex plane. While we do not know yet how to get $H(G)$ from the spectrum of $L$,
we started to study what happens if the wave amplitude $h(x) \in \mathbb{K}$ 
is deformed at a single simplex $x \in G$ and kept constant everywhere else. This is studied
in part (D). 

\section{Result (C): Energy theorem}

\paragraph{}
Also the generalization of the energy theorem needs that $G$ is a simplicial complex,
a finite set of non-empty sets closed under the operation of taking non-empty subsets. 
It has already been formulated in the complex case as a remark in \cite{KnillEnergy2020}, 
but it holds in general. The reason is that both sides do not really tap into the 
multiplicative structure of the algebra. 

\begin{thm}[Energy theorem]
Assume $G$ is a finite abstract simplicial complex. 
For any $h: G \to \mathbb{K}$, we have the energy relation $\sum_{x,y} g(x,y) = H(G)$.
\label{3}
\end{thm}

\paragraph{}
We can establish the statement by standing on the shoulders of the theorem in the topological
case \cite{KnillEnergy2020}, and just comment on the later.

\begin{proof}
We just note that both sides of the equation are multi-affine in each energy value entry 
$X=h(x)$. This means that if change the single entry $X$, then the left hand side is of the 
form $a + b X$ with constants $a,b \in \mathbb{K}$ and the right hand side is $c + d X$
again with constants $c,d \in \mathbb{K}$. 
Then we notice that we know the relation in the constant zero case $h=0$, 
(where both sides are zero) and for $h(x)=\omega(x)$, where the theorem has been 
proven already and where both sides are the Euler characteristic. The
term $H(G) = \sum_{x \in G} h(x)$ obviously even linear in each of the entries so that $a+bX=X$.
Also, each term $g(x,y) = \omega(x) \omega(y) H(W^+(x) \cap W^+(y))$ is an affine function in each
entry $X=h(x_0)$. The sum $\sum_{x,y} g(x,y) = H(G)$ therefore is also affine. 
Having the values agree on two points assures us now that $c=a=0$ and $b=d$. 
\end{proof}

\paragraph{}
Let us just remind about the proof of Theorem~(\ref{2}) in the case
$h(x)=\omega(x)$. The proof itself does not directly generalize to complex valued fields. 
It has the following ingredients: \\ 

\begin{itemize}
\item 
Assume $h(x)=\omega(x)$ so that $H(G) = \chi(G)$ is the Euler characteristic. 
If $x \in G$, then $S(x) = W^-(x) + W^+(x)$, where $+$ is the join or  Zykov addition of 
simplicial complexes, which is dual to the disjoint union. The compatibility of the genus
$1-H(A)$ with the join has now the consequence that
$$    1-H(S(x)) = (1-H(W^-(x))) (1-H(W^+(x))  \; .   $$
This implies with $(1-H(W^-(x))) = \omega(x0$ 
that  $H(G) = \sum_x (1-H(W^+(x))) = \sum_x  \omega(x) (1-H(S(x))$ 
which we consider as a Gauss-Bonnet relation for Euler characteristic $H=\chi$. 

\item The above relation shows that the {\bf super trace} of the matrix $g$ which
$\sum_{x} \omega(x) g(x,x) =\sum_x K(x)$ agrees with the {\bf total energy} $H(G)=\chi(G)$.
This {\bf Gauss-Bonnet relation} which follows from a
{\bf Poincar\'e-Hopf relation} for the valuation $H=\chi$ defined by $h$. 
We can think about $\omega(x) g(x,x)$ as a curvature.

\item The last observation is that the potential $V(x) = \sum_{y \in G} g(x,y)$
which leads to the potential theoretical energy 
$$  V(x) = \sum_{x,y} g(x,y)  $$
of the vertex $x$ satisfies $V(x) = \omega(x) g(x,x)$. This shows that
{\bf curvature} of $x$ is equal to {\bf potential energy} of $x$ induced from 
all other simplices, (including self-interaction).
\end{itemize}

\paragraph{}
We are excited about the set-up because in classical physics, {\bf self-interaction}
is a sensitive issue. When looking at the electric field of a bunch of electrons, then 
at the point of each of the electrons, we have to disregard the field of the electron
itself, as it is infinite. To cite Paul Dirac from an interview given in 1982: 
{\it "I think that the present methods which theoretical physicists are using are not the correct methods.
They use what they call a renormalization technique, which involves handling infinite quantities.
And this is not mathematically a logical process. I would say that it is just
a set of ``working rules" rather than a correct mathematical theory. 
I don't like this whole development at all. I think that some other important 
discoveries will have to be made, before these questions are put into order."}
People always overestimate their own work, but we just want to point out that the 
field theory on finite set of sets $G$ taking values in $\mathcal{K}$ as worked on 
in the current document is completely absent from any infinities! 

\section{Result (D): Geometric phase}

\paragraph{}
For the last part, we assume $\mathbb{K}=\mathbb{C}$ as we can look at the spectrum of $L$.
We have seen in the real case, that the number of positive minus the number of negative
eigenvalues of $L$ is then $\chi(G)= H(G)$ as noticed in 2017 \cite{HearingEulerCharacteristic}.
We can try to generalize this.  While the spectrum of quaternion matrices is always 
non-empty \cite{Wood1984} which is related to the fundamental theorem of algebra 
in division algebras, there are indications that we can not always assign to 
$x \in G$ a canonical eigenvalue $\lambda_k$. 
 Such difficulties is the reason that we assume here $\mathbb{K}=\mathbb{C}$. We still believe
that the spectral situation in the quaternion and octonion case should be studied more.
The reason is that the complex case physically just looks too much ``electromagnetic" only. 

\paragraph{}
Let us look at a circle $t \to h_{t,x}(y) = h(y)$ if $y \in G$ is not in $G$ and
$h_{t,x}(x)= e^{i t} h(x)$.  Let $\lambda_k(h_{t,x})$ be the eigenvalues of $L$
defined by $G$ and the energy $h_{z,x}$. Let $W(\lambda_k)$ denote the winding number of
the path $t \to \lambda_k(h_{t,x})$. This is well defined because the
eigenvalues are never zero by the determinant formula. 

\begin{thm}[Geometric phase]
For every $x$, a circular deformation of the value $h_t(x) = h(x) e^{i t}$ produces a 
in general a {\bf nontrivial} {\bf permutation} of the eigenvalues, 
when $t$ goes around a circle from $0$ to $2\pi$. 
\label{4}
\end{thm}

\paragraph{}
The proof of Theorem~(\ref{4}) is an explicit computation,
in a concrete situation. An example is $G=K_3$, the set of all non-empty
subsets of $\{1,2,3\}$. The simplicial complex $G$ contains the $7$ sets
$$ \{\{1\},\{2\},\{3\},\{1,2\},\{1,3\},\{2,3\},\{1,2,3\}\} \; . $$
We can take the energy values $h(x_k)=e^{2 \pi i k/7}$,
where $k$ goes from $1$ to $7$. These are the $7$'th {\bf root of unities}.

\begin{figure}[!htpb]
\scalebox{0.8}{\includegraphics{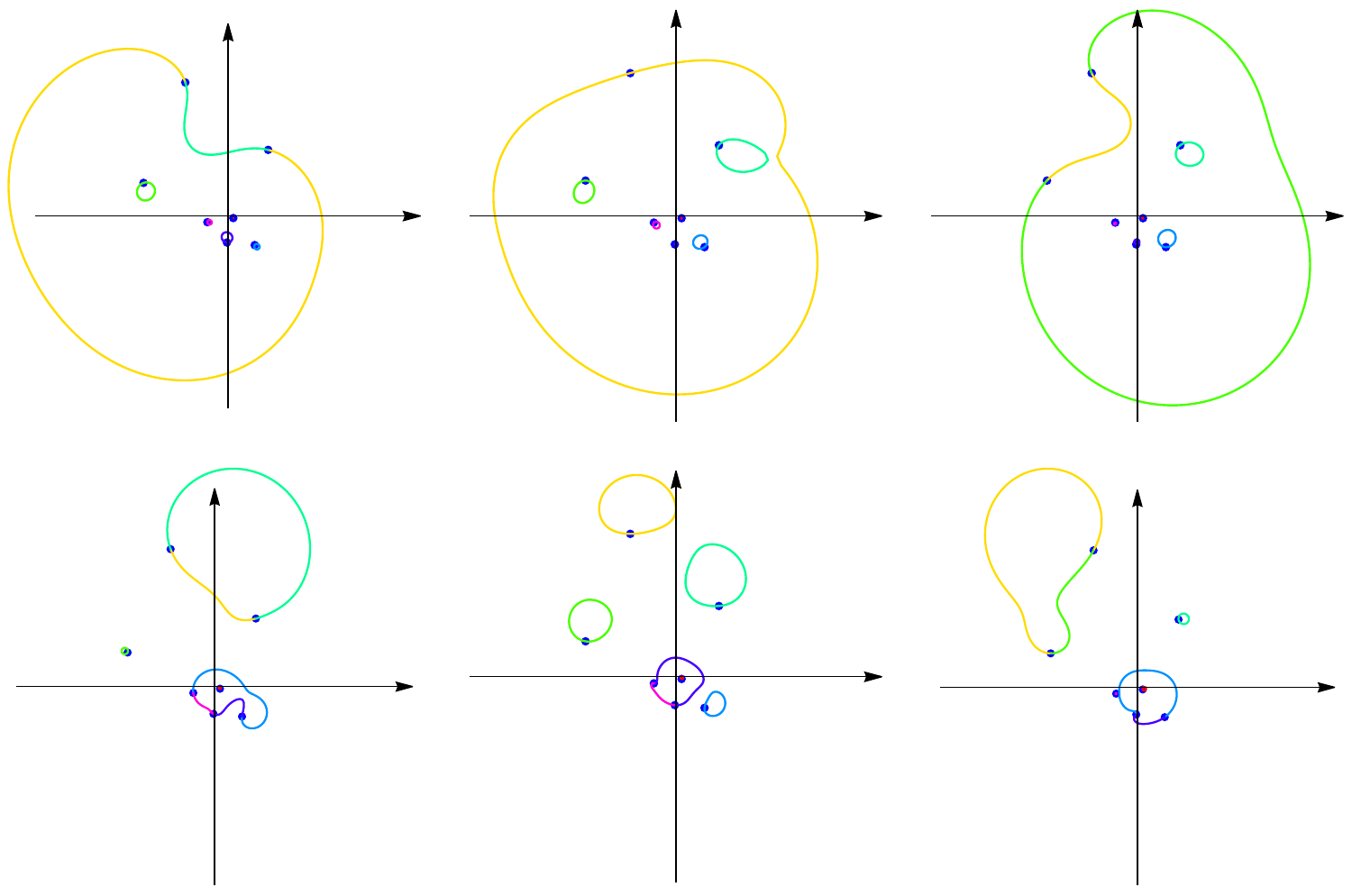}}
\label{spectralcurves}
\caption{
For $G=K_3$ and $h(x)$ taking the $7$'th roots of unity, we generate
the group $\Pi(G,h)= S_3 \times S_3$ which has $36$ elements. 
The figure shows the moves when turning wheels $1$ to $6$ by $360$ 
degrees each. 
}
\end{figure}

\paragraph{}
Remark: In the real case, we most of the time have a natural map
$\Psi: h  \in (\mathbb{R}^*)^n \to \lambda \in (\mathbb{R}^*)^n$,
if  $\lambda_k$ is the eigenvalue which has the property that the
circle $\lambda_k(\theta)$ under the deformation $h_\theta(x) = h(x) e^{i \theta}$
has non-zero winding number with respect to the origin $0 \in \mathbb{C}$.
It can however happen that two rotations produce the same deformation 
of the spectrum. It is still possible, when building up $G$ to associate
to each $x \in G$ a unique eigenvalue $\lambda_k(L)$. 

\paragraph{}
If we look at the one-parameter circle of energy functions $t \to h_{t,x}$,
the winding numbers are all integers which because the depend continuously on 
parameters are constant on the set of all energy functions. When deforming from 
the real case, then we can not hit a real eigenvalue until $\theta=\pi$. 
What is possible however is that if we deform a single $h(x)$ value along a circle
is that two eigenvalues turn around the origin. 
We initially thought that this is not possible. In which cases this is possible has
still to be investigated. 

\section{Result (E): Complex manifolds}

\paragraph{}
The {\bf general linear group} $GL(n,\mathbb{C})$ can be identified with 
an open subset of $\mathbb{C}^{n^2}$. 
It so is naturally a non-compact {\bf K\"ahler manifold}, because every complex sub-manifold
of a K\"ahler manifold with induced complex structure is K\"ahler. In our case, we have 
an explicit parametrization
$$   h \in \mathbb{C}^{n}  \to {\rm Symm}(n,\mathbb{C}) = \mathbb{C}^{n(n+1)/2} $$
of an $n$-dimensional complex manifold $M$ which has the property that 
the multiplicative subgroup $(\mathbb{C}^*)^n$ is mapped into 
$GL_{{\rm symm}}(n,\mathbb{C})$, the manifold space of symmetric complex 
$(n \times n)$-matrices.

\paragraph{}
Unlike manifolds of self-adjoint matrices, spaces of symmetric matrices are always K\"ahler
manifolds. In our case, the parametrization map $r: h \to L(h)$ is {\bf multi-linear}. 
Its {\bf rank} is the rank of the matrix $dr^T dr$ which is an integer only depending on $G$. 
Indeed, it agrees with the rank of $G$. 
We measure for positive dimensional simplicial complexes that the
determinant of the K\"ahler metric ${\rm det}(dr^T dr)$ is
divisible by $3$. It is clearly $1$ if $G$ is zero dimensional as then $r$ maps $h$ into a 
diagonal matrix. For complete complexes, we have the following ranks
$G=K_1$ has rank $3^0$, $G=K_2$ has rank $3^2$, $G=K_3$ has rank $3^9$ and $G=K_4$ has rank $3^{15}$. 
Looking up the integer sequence we expect the rank
for $G=K_{n+1}$ to be $a_n=3^{\sum_{k=1}^{n-1} B(n,k) (n-k)}$. It should be possible to 
prove this by induction. We have not yet done so. 

\begin{comment}
K2  :   3^2      3
K3  :   3^9      7
K4  :   3^28    15
K5  :   3^75    31
K6  :   3^186   63
K7  :   3       127
a[n_]:=Sum[ Binomial[n,k] (n-k), {k,n-1}];
\end{comment}

\paragraph{}
We can now look at the open submanifold $S$ of $GL_{{\rm symm}}(n,\mathbb{C})$ which consists
of matrices which have {\bf simple spectrum}. This is still a non-compact 
complex manifold of complex dimension $n(n+1)/2$. As the collision sets are of smaller dimension,
a random complex symmetric matrix is in $S$. (This is much less obvious in the real case 
\cite{TaoVu2017}). 
The K\"ahler manifold $S$ is dense, is connected and simply connected. 
As matrices with simple spectrum are diagonalizable and in general matrices over $\mathbb{K}=\mathbb{C}$
can be put into a {\bf Jordan normal form}, this 
follows from the connectedness and simply connectedness of the unitary group 
$U(n,\mathbb{C})$ for $n>1$. 

\paragraph{}
When looking at connection matrices $L(G,h)$, then we have $n$-dimensional complex 
manifold of matrices, also if we intersect it with $S$. We get then 
a complex sub-manifold of the K\"ahler manifold $S$ and is so K\"ahler when taking 
the induced complex structure. We actually have an explicit parametrization
$r: h=(h_1, \dots, h_n) \to L(G,h)$ and so also explicit coordinates and an
explicit K\"ahler bilinear form $g(h)=dr^T(h) \cdot dr(h)$, where $dr(h)$ is the {\bf Jacobian}
of $r: \mathbb{C}^n \to \mathbb{C}^{n(n+1)/2}$ at the point $h$. For a fixed set of sets $G$,
the open manifold $M=r((\mathbb{C}^*)^n) \cap \subset S$ consists of {\bf connection
Laplacians} which have simple spectrum. Here $\mathbb{C}^* = \mathbb{C} \setminus \{0\}$ is
the {\bf multiplicative group} in the field $\mathbb{C}$ of complex numbers.

\paragraph{}
We can now make non-trivial statements about the manifold $M$. 
In the $0$-dimensional case where $M = \{ x_1=\{ v_1 \}, \dots , x_n=\{ v_n \} \}$, 
the manifold $M$ is an open sub manifold $M=(\mathbb{C}^*)^n)$ consisting of all
vectors, for which all coordinates are non-zero and different. This manifold is connected
but not simply connected. The fundamental group is $\pi_1(M) = \mathbb{Z}^n$ everywhere.
It is a bit surprising that in in general, when the dimension of $G$ gets bigger,
the manifold $M(G,h)$ can have a non-commutative 
fundamental group and that it is not connected, as different components can have different
fundamental groups. We can compute them explicitly.

\paragraph{}
If 
$$ \Pi(G,h_0) = \{ (g_1,...,g_n), g_k^{n_k}=1, g_i^{n_i} g_j^{n_j} g_i^{-n_i} g_j^{-n_j} = 1 \} $$ 
is the {\bf finite presentation} of symmetry group of the spectrum defined by the above 
deformations, define the now infinite but still {\bf finitely presented group}
$$ \pi(G,h_0) = \{ (g_1,\dots ,g_n), g_i^{n_i} g_j^{n_j} g_i^{-n_i} g_j^{-n_j} = 1 \} $$
which is the free group with $n$ generators in which only the {\bf mixed relations} of $\Pi(G,h_0)$ are
picked as the relations. In the case when $\Pi(G,h_0)$ is the trivial group with one element, 
then $\pi(G,h_0)$ is the Abelian group $\mathbb{Z}^n$. 

\begin{thm}
The manifold $M=M(G)$ is in general not connected. The fundamental group
of the connected component containing $L(G,h_0)$ is the group $\pi(G,h_0)$. 
\label{(5)}
\end{thm}

\paragraph{}
\begin{proof}
To see that the manifold is in general not connected,
take a fixed complex like $G=\mathbb{K}_3$, and notice that there can be
different vectors $h=(a,b,c) \in (\mathbb{C}^*)^2$ for which the groups
are different. We have included Mathematica code which allows to verify this numerically.
[To prove this mathematically, one would have to establish a computer assisted
proof using interval arithmetic, establishing that the deformations really do
what we see. It might be simpler to actually understand this theoretically more 
first and understand {\bf why} the eigenvalues get permuted at all.]
The groups are obviously constant on each connected component,
so that $M$ must have different connectivity components. 
To see that $\pi(G,h_0)$ is the fundamental group, we note that every closed
curve in $M$ can be described as a closed curve in the parameter domain 
$(\mathbb{C}^*)^n$  but that unlike in that parameter manifold 
which is homotopic to the $n$-torus $\mathbb{T}^n$, the 
image has now a more complicated topology. \\
Let $\gamma$ be a closed path in $M$ starting at $h_0$ which is obtained
from a generator, where the $h$ value is turned at a single simplex $x \in G$.
It is not possible that some multiple $\gamma + \cdots + \gamma$ of the curve
(doing the loop several times) is homotopic to a point because such a 
deformation would produce a deformation 
$r^{-1}(M) = (\mathbb{C}^*)^n$ which is homotopic to $\mathbb{T}^n$. This is
not possible for the linear map $r$.
\end{proof}

\paragraph{}
Note that if $G$ is $0$-dimensional, the finite permutation group $\Pi(M,h_0)$ is 
trivial for all $h_0 \in  (\mathbb{C}^*)^n$; it has only one element, the identity.
When looking at this from the point of view of finitely presented groups, 
then the fundamental group with generators $(g_1, \dots, g_n)$ has then all the
{\bf pair relations} $r_{ij} = g_j g_k g_j^{-1} g_k^{-1}$ so that the fundamental group is 
$\mathbb{Z}^n = \{ (g_1, \dots ,g_n) | (r_{12}, \dots ,r_{(n-1),n}) \}$ 
which is the free group with generators $g_1, \dots, g_n$ modulo these relations. 
In other words, the fundamental group $\mathbb{Z}^n$ is then the Abelianization of the 
free group with $n$ generators.  In general, some of these pair relations become 
more complicated. 

\section{Examples}

\paragraph{}
If $G=K_2=\{ \{1,2\},\{1\},\{2\} \}$ and
$h=[U,V,W]$ are the energies (field values) in $\mathcal{K}$, we
have the Study determinant ${\rm det}(L)=|U V W| = {\rm det}(g)$ and
$$
L=  \left[
                  \begin{array}{ccc}
                   U & 0 & U \\
                   0 & V & V \\
                   U & V & U+V+W \\
                  \end{array}
                  \right],  
g = \left[
                  \begin{array}{ccc}
                   U+W & W & -W \\
                   W & V+W & -W \\
                   -W & -W & W \\
                  \end{array}
                  \right] $$
so that
$$ \overline{g} L    =   \left[
                  \begin{array}{ccc}
                   |U|^2 & 0 & |U|^2-|W|^2 \\
                   0 & |V|^2 & |V|^2-|W|^2 \\
                   0 & 0 & |W|^2 
                  \end{array}
                  \right] \; . $$

\paragraph{} 
For $G=\{\{1\},\{2\},\{3\},\{1,2\},\{2,3\}\}$ and $h=[U,V,W,P,Q]$,
the study determinant of $L$ and $g$ is $|U,V,W,P,Q]$. Then 
$$  L=\left[
                  \begin{array}{ccccc}
                   U & 0 & 0 & U & 0 \\
                   0 & V & 0 & V & V \\
                   0 & 0 & W & 0 & W \\
                   U & V & 0 & P+U+V & V \\
                   0 & V & W & V & Q+V+W \\
                  \end{array}
                  \right] $$
$$  g=\left[
                  \begin{array}{ccccc}
                   P+U & P & 0 & -P & 0 \\
                   P & P+Q+V & Q & -P & -Q \\
                   0 & Q & Q+W & 0 & -Q \\
                   -P & -P & 0 & P & 0 \\
                   0 & -Q & -Q & 0 & Q \\
                  \end{array}
                  \right]  $$
so that 
$$ \overline{g} L  =  \left[ 
                  \begin{array}{ccccc}
                   |U|^2 & 0 & 0 & |U|^2-|P|^2 & 0 \\
                   0 & |V|^2 & 0 & |V|^2-|P|^2 & |V|^2-|Q|^2 \\
                   0 & 0 & |W|^2 & 0 & |W|^2-|Q|^2 \\
                   0 & 0 & 0 & |P|^2 & 0 \\
                   0 & 0 & 0 & 0 & |Q|^2 \\
                  \end{array}  \right]  \; . $$

\paragraph{} 
For $G=\{\{1\},\{2\},\{3\},\{1,2\},\{2,3\}\}$ and $h=[U,V,W,P,Q]$,
the Study determinant of $L$ and $g$ is $|UVWPQ|$. Then
$$  L=\left[
                  \begin{array}{ccccc}
                   U & 0 & 0 & U & 0 \\
                   0 & V & 0 & V & V \\
                   0 & 0 & W & 0 & W \\
                   U & V & 0 & P+U+V & V \\
                   0 & V & W & V & Q+V+W \\
                  \end{array}
                  \right] $$
$$ g=\left[
                  \begin{array}{ccccc}
                   P+U & P & 0 & -P & 0 \\
                   P & P+Q+V & Q & -P & -Q \\
                   0 & Q & Q+W & 0 & -Q \\
                   -P & -P & 0 & P & 0 \\
                   0 & -Q & -Q & 0 & Q \\
                  \end{array}
                  \right]  $$
so that
$$ \overline{g} L  =  \left[ 
                  \begin{array}{ccccc}
                   |U|^2 & 0 & 0 & |U|^2-|P|^2 & 0 \\
                   0 & |V|^2 & 0 & |V|^2-|P|^2 & |V|^2-|Q|^2 \\
                   0 & 0 & |W|^2 & 0 & |W|^2-|Q|^2 \\
                   0 & 0 & 0 & |P|^2 & 0 \\
                   0 & 0 & 0 & 0 & |Q|^2 \\
                  \end{array}  \right]  \; . $$

\paragraph{}
For $G=\{\{1\},\{2\},\{3\},\{4\},\{1,2\},\{2,3\},\{2,4\},\{3,4\},\{2,3,4\}\}$
and $h=[P,Q,R,U,V,W,X,Y,Z]$ the determinant is $P Q R U V W X Y Z$. Then
\begin{tiny}
$$ L = \left[
                  \begin{array}{ccccccccc}
                   U & 0 & 0 & 0 & U & 0 & 0 & 0 & 0 \\
                   0 & V & 0 & 0 & V & V & V & 0 & V \\
                   0 & 0 & W & 0 & 0 & W & 0 & W & W \\
                   0 & 0 & 0 & P & 0 & 0 & P & P & P \\
                   U & V & 0 & 0 & Q+U+V & V & V & 0 & V \\
                   0 & V & W & 0 & V & R+V+W & V & W & R+V+W \\
                   0 & V & 0 & P & V & V & P+V+X & P & P+V+X \\
                   0 & 0 & W & P & 0 & W & P & P+W+Y & P+W+Y \\
                   0 & V & W & P & V & R+V+W & P+V+X & P+W+Y & P+R+V+W+X+Y+Z \\
                  \end{array}
                  \right] $$
\end{tiny}
and 
\begin{tiny}
$$ g = \left[ 
                  \begin{array}{ccccccccc}
                   Q+U & Q & 0 & 0 & -Q & 0 & 0 & 0 & 0 \\
                   Q & Q+R+V+X+Z & R+Z & X+Z & -Q & -R-Z & -X-Z & -Z & Z \\
                   0 & R+Z & R+W+Y+Z & Y+Z & 0 & -R-Z & -Z & -Y-Z & Z \\
                   0 & X+Z & Y+Z & P+X+Y+Z & 0 & -Z & -X-Z & -Y-Z & Z \\
                   -Q & -Q & 0 & 0 & Q & 0 & 0 & 0 & 0 \\
                   0 & -R-Z & -R-Z & -Z & 0 & R+Z & Z & Z & -Z \\
                   0 & -X-Z & -Z & -X-Z & 0 & Z & X+Z & Z & -Z \\
                   0 & -Z & -Y-Z & -Y-Z & 0 & Z & Z & Y+Z & -Z \\
                   0 & Z & Z & Z & 0 & -Z & -Z & -Z & Z \\
                  \end{array} \right]  \; . $$
\end{tiny}
Then 
\begin{tiny}
$$ \overline{g} L = \left[
                  \begin{array}{ccccccccc}
                   |U|^2 & 0 & 0 & 0 & |U|^2-|Q|^2 & 0 & 0 & 0 & 0 \\
                   0 & |V|^2 & 0 & 0 & |V|^2-|Q|^2 & |V|^2-|R|^2 & |V|^2-|X|^2 & 0 & -|R|^2+|V|^2-|X|^2+|Z|^2 \\
                   0 & 0 & |W|^2 & 0 & 0 & |W|^2-|R|^2 & 0 & |W|^2-|Y|^2 & -|R|^2+|W|^2-|Y|^2+|Z|^2 \\
                   0 & 0 & 0 & |P|^2 & 0 & 0 & |P|^2-|X|^2 & |P|^2-|Y|^2 & |P|^2-|X|^2-|Y|^2+|Z|^2 \\
                   0 & 0 & 0 & 0 & |Q|^2 & 0 & 0 & 0 & 0 \\
                   0 & 0 & 0 & 0 & 0 & |R|^2 & 0 & 0 & |R|^2-|Z|^2 \\
                   0 & 0 & 0 & 0 & 0 & 0 & |X|^2 & 0 & |X|^2-|Z|^2 \\
                   0 & 0 & 0 & 0 & 0 & 0 & 0 & |Y|^2 & |Y|^2-|Z|^2 \\
                   0 & 0 & 0 & 0 & 0 & 0 & 0 & 0 & |Z|^2 \\
                  \end{array} \right] \; . $$
\end{tiny}
This illustrates for example ${\rm tr}(\overline{g} L) = \sum_x |h(x)|^2$.

\begin{comment}
Generate[A_]:=Delete[Union[Sort[Flatten[Map[Subsets,A],1]]],1];
G=Sort[{{1},{2},{3},{4},{1,2},{2,3},{2,4},{3,4},{2,3,4}}];n=Length[G];    e={U,V,W,P,Q,R,X,Y,Z};
G=Sort[{{1},{2},{3},{1,2},{2,3}}];n=Length[G];    e={U,V,W,P,Q};
S=Table[-(-1)^Length[G[[k]]]*If[k ==l,1,0],{k,n},{l,n}]; (* Super matrix   *)
energy[A_]:=If[A=={},0,Sum[e[[Position[G,A[[k]]][[1,1]]]],{k,Length[A]}]];
star[x_]:=Module[{u={}},Do[v=G[[k]];If[Q[v,x],u=Append[u,v]],{k,n}];u];
core[x_]:=Module[{u={}},Do[v=G[[k]];If[Q[x,v],u=Append[u,v]],{k,n}];u];
Wminus     = Table[Intersection[core[G[[k]]],core[G[[l]]]],{k,n},{l,n}];
Wplus      = Table[Intersection[star[G[[k]]],star[G[[l]]]],{k,n},{l,n}];
Wminusplus = Table[Intersection[star[G[[k]]],core[G[[l]]]],{k,n},{l,n}];
Wplusminus = Table[Intersection[core[G[[k]]],star[G[[l]]]],{k,n},{l,n}];
Lminus     = Table[energy[Wminus[[k,l]]],    {k,n},{l,n}];  L  =      Lminus;
Lplus      = Table[energy[Wplus[[k,l]]],     {k,n},{l,n}];  g  =   S.Lplus.S;
Potential  = Table[Sum[g[[k,l]],{l,n}],{k,n}];
Superdiag  = Table[S[[k,k]]*g[[k,k]],{k,n}];
Total[Flatten[g]] - Total[e]
Tr[S.g] -Total[e]
{Det[g],Det[L],Product[e[[k]],{k,n}]}
Superdiag -Potential
g.L
\end{comment}

\paragraph{}
Let $G=\{\{1\},\{1,3,4\},\{1,4,5\},\{4\},\{1,4\}\}$. It is a set of sets 
and not a simplicial complex. It is also not ordered as we usually do. 
Now take the energy values $h={2,4,3,-1,X}$, where $X$ is a variable. We want
to illustrate the determinant formula, but not assume any associativity
or commutativity. When looking at the Leibniz determinant we get
a complicated expression which simplifies in this case where only one 
variable appears to $-24 X$.

\paragraph{}
For $G=\{\{1,3,4\},\{4\}\}$ with $h=(1,X)$ we have
$$ L = \left[
                  \begin{array}{cc}
                   X+1 & X \\
                   X & X \\
                  \end{array}
                  \right], 
  g = \left[
                  \begin{array}{cc}
                   1 & 1 \\
                   1 & X+1 \\
                  \end{array}
                  \right] $$
The computation
$$ g.L = \left[
                  \begin{array}{cc}
                   2 X+1 & 2 X \\
                   (X+1)^2 & X (X+2) \\
                  \end{array}
                  \right] $$ 
shows that even if $X=1$, we do not have a diagonal matrix.

\paragraph{}
For the simplicial complex
$G=\{\{1,2\},\{2,3\},\{1\},\{2\},\{3\}\}$, which is
ordered from larger dimensions to smaller dimensions and
$h=(1,1,1,1,X)$, we have
$$ L =  \left[ 
                  \begin{array}{ccccc}
                   3 & 1 & 1 & 1 & 0 \\
                   1 & X+2 & 0 & 1 & X \\
                   1 & 0 & 1 & 0 & 0 \\
                   1 & 1 & 0 & 1 & 0 \\
                   0 & X & 0 & 0 & X \\
                  \end{array} \right ] \; $$
and 
$$ g = \left[
                  \begin{array}{ccccc}
                   1 & 0 & -1 & -1 & 0 \\
                   0 & 1 & 0 & -1 & -1 \\
                   -1 & 0 & 2 & 1 & 0 \\
                   -1 & -1 & 1 & 3 & 1 \\
                   0 & -1 & 0 & 1 & X+1 \\
                  \end{array}
                  \right] \; , $$
so that
$$ g L =  \left[
                  \begin{array}{ccccc}
                   1 & 0 & 0 & 0 & 0 \\
                   0 & 1 & 0 & 0 & 0 \\
                   0 & 0 & 1 & 0 & 0 \\
                   0 & 0 & 0 & 1 & 0 \\
                   0 & X^2-1 & 0 & 0 & X^2 \\
                  \end{array}
                  \right]  \; . $$
Since the order was up to down, the matrix $g L$
is lower triangular. 

\section{Illustrations}

\begin{figure}[!htpb]
\scalebox{0.6}{\includegraphics{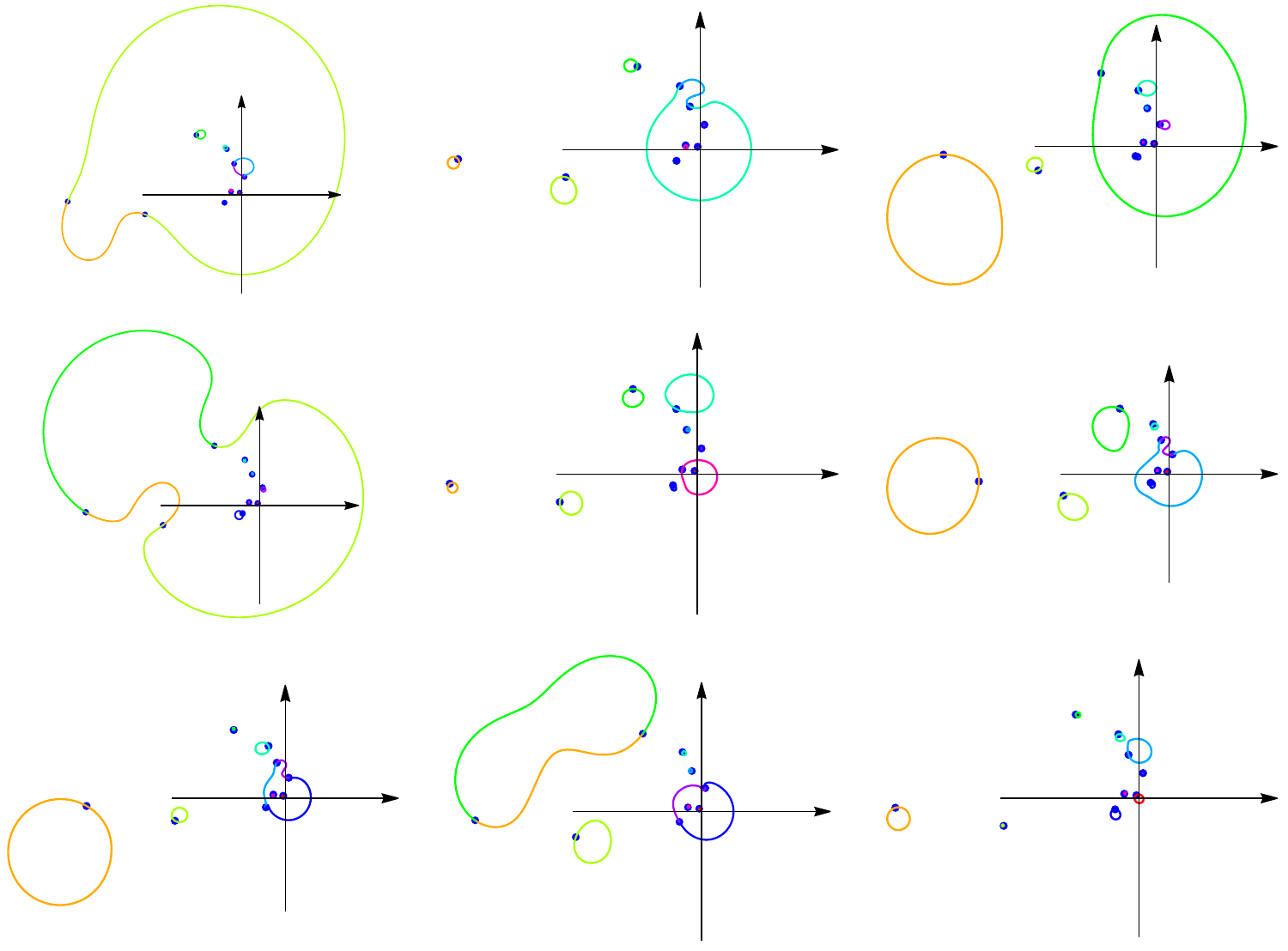}}
\label{spectralcurves}
\caption{
The spectral curves $\theta \to \lambda_k(G,h_{x,\theta})$ are seen
in the complex plane $\mathbb{C}$, when the energy $h(x)$ is rotated around
by $2\pi$ at the simplex $x$. In this case, $n=|G|=9$. For each $x$, we see here
exactly one eigenvalue that circles around the origin. If this is the case,
we can attach to each simplex a unique spectral value $\lambda(x)$
and keep track of this correspondence even if there should be spectral
collisions. The spectrum of $L$ is a priori only a set. 
}
\end{figure}

\begin{figure}[!htpb]
\scalebox{0.6}{\includegraphics{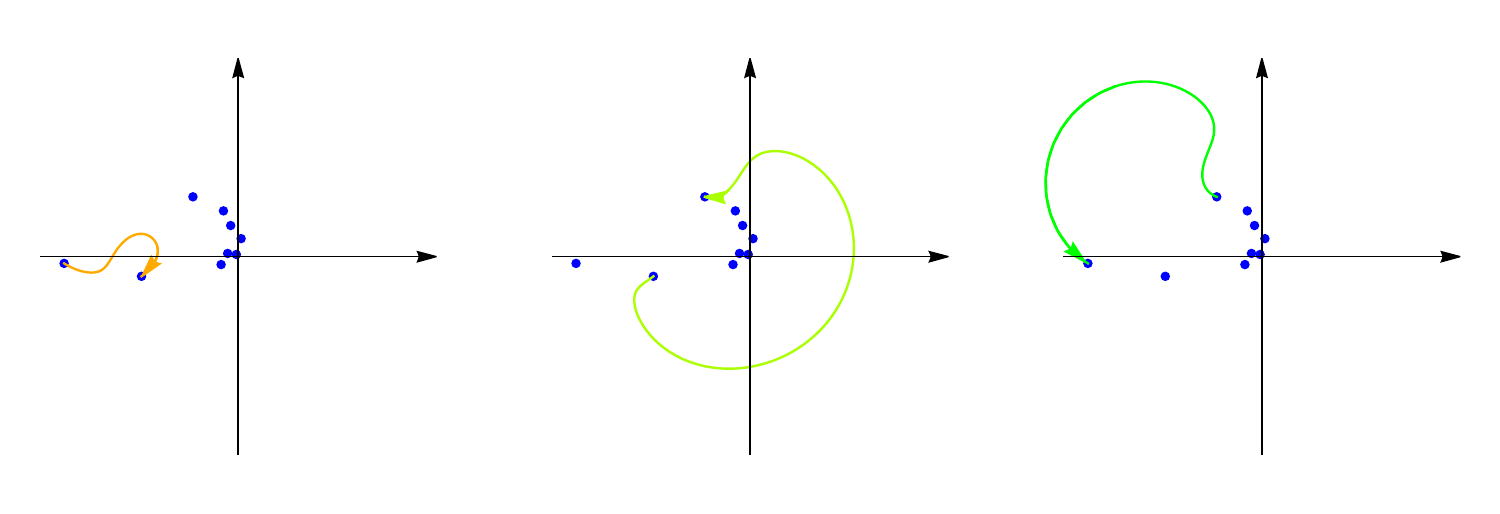}}
\label{geometricphase}
\caption{
Deforming the function value $h(x)$ at some simplex $x$ by
multiplying with $e^{it}$ produces a deformation of the eigenvalues.
It is a bit surprising  as we would have expected to have each eigenvalue
move on a circle.  Here is an example, where a cyclic permutation 
of $3$ eigenvalues appears. This means that we can assign to that simplex $x$ not 
just one particle (eigenvector) but that 3 particles are tied together
by symmetry.  We have so far not seen this when $h$ takes values in the real case.
}
\end{figure}

\begin{figure}[!htpb]
\scalebox{0.8}{\includegraphics{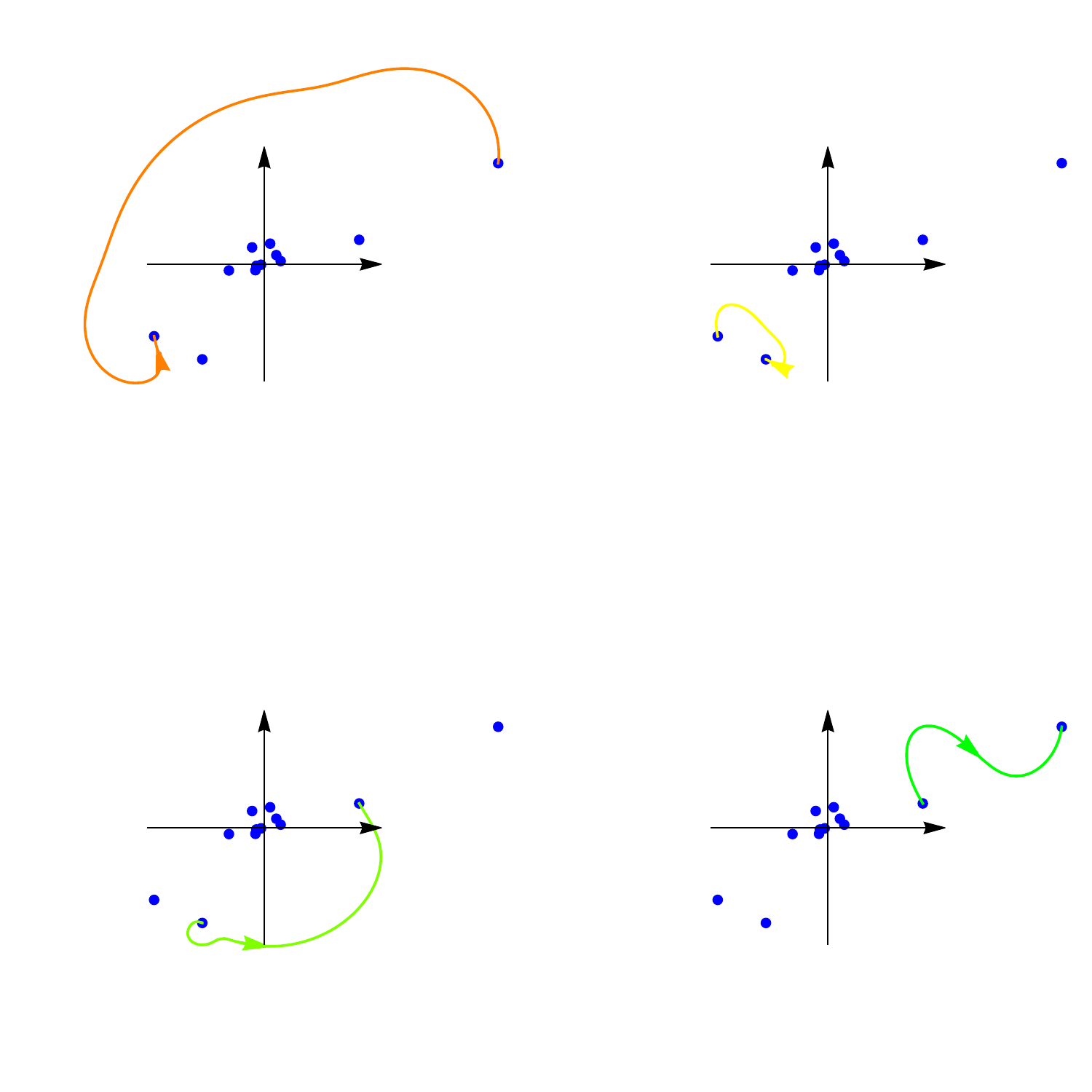}}
\label{geometricphase2}
\caption{
A case with $h$ taking values in $\mathbb{U}=\mathbb{T}$ in 
the $\mathbb{K}=\mathbb{C}$ case. The deformation 
of one of the $h$-values permutes $4$ eigenvalues. Together
with other wheels turning, we can realize the full permutation
group of these four eigenvalues. 
}
\end{figure}

\section{Discussion}

\paragraph{}
In physics, a complex-valued function $\psi: M \to \mathbb{C}$ on
a geometry is also seen as a {\bf wave} because quantum mechanics is
also wave mechanics. A $\mathbb{T}$-valued function can be seen
as the section of a fibre bundle for which the gauge group is the
circle $\mathbb{T}$. In the case of quaternions, one has a non-Abelian gauge
field situation like $\mathbb{U}=SU(2)$. We still have to see, 
whether one can get any type of physical content from the theorems. 
We can look at the {\bf Schr\"odinger} equation $i \hbar \psi_t = H \psi_t$ 
for example for some self-ajoint $H$ and then feed in $h(x) = \psi_t(x)$, 
producing an operator $L_t$ with spectrum and K\"ahler manifolds $M_t$. 
Can we have situations where the motion of the wave produces topology changes
of $M_t$?

\paragraph{}
Motivated from {\bf quantum mechanics} and {\bf gauge field theories} 
it was natural to look at functions on 
a finite abstract simplicial complex $G$ by assigning to each of the $n$ simplices $x$
of $G$ an element $h(x)$ of a division or operator algebra $\mathbb{K}$ and look at the corresponding
connection Laplacian $L$. Especially interesting should be the case when we use unitary operators
that is if $h$ is $\mathbb{U}$-valued. 
More generally, motivated maybe also from number theory, one can look at any ring $\mathbb{K}$, 
commutative or not and look at $\mathbb{K}$-valued functions. 
Note however that unlike for gauge field theories, where only the multiplicative structure of
a Lie group is used, the field object $(G,h)$ uses both the additive as well as the 
multiplicative structure of the algebra $\mathbb{K}$. 

\paragraph{}
When using $\mathbb{C}$-valued energy values $h(x)$, the map $G \to \mathbb{C}$
can also be seen as a ``quantum wave". The $n \times n$ connection matrix $L$ that is built 
from the wave has now a spectrum. 
The map $h \to L$ is an explicit linear parametrization of a complex manifold 
of $n \times n$ matrices for which the determinant is the product of the $h$ values.
For $0$-dimensional complexes, we just have $(h_1,\dots, h_n) \to L={\rm Diag}(h_1,\dots,h_n)$. 
What happens here is that for every wave $h$ we have operators attached, similarly as in 
quantum field theories. Now, it would be interesting to know whether the spectrum of $L(G,h)$
has any physical content. As a wave now also defines complex manifolds we should
expect the manifolds to not change fundamentally in topology when the wave evolves. Or then
that something dramatic changes if the manifolds change topology. 

\paragraph{}
If the ring $\mathbb{K}$ is non-commutative, like for
quaternion $\mathbb{H}$-valued fields, one also needs to work with a fixed
order on $G$. Geometry becomes non-commutative. 
More references on linear algebra in non-commutative settings 
is contained in \cite{Rodman2014,Voight2020,Aslaksen1996}.
The Leibniz determinant is $\mathbb{K}$-valued), the 
Study determinant \cite{Study1920,CohenDeLeo2000} is
real valued and the 
Dieudonn\'e determinant \cite{Dieudonne1943,ArtinGeometricAlgebra,RosenbergKTheory,Brenner1968}
takes values in the Abelianization $\overline{\mathbb{K}}$ of $\mathbb{K}$. 
This works especially for quaternions $\mathbb{K}=\mathbb{H}$.
In the case of octonions $\mathbb{K}=\mathbb{O}$, which is no more a ring, 
the non-associativity complicates the linear algebra a bit more. 
There had been early physical motivation for non-associative structures
like \cite{JordanNeumannWigner}. The Octonions have been at various places been 
seen associated with gravity. These ideas are exciting \cite{Baez2002,
Gursey,DixonDivisionalgebras,ConwaySmith,Baez2002,Furey2016,RowlandsRowlands}.

\paragraph{}
An other goal of this note was to point out a remarkable 
{\bf geometric spectral phase phenomenon} in the case $\mathbb{K}=\mathbb{C}$,
where one can use the usual determinant and where one also has well defined eigenvalues 
and eigenvectors.  Given a wave $h: G \mathbb{C}$, we can define a non-trivial 
map from the fundamental group of $\mathcal{U}^n = \mathbb{T}^n$ to the permutation 
group of the spectrum of $L$. Think of turning a wheel at a simplex $x$ and 
rotating the value of the field $h(x)$ there along a circle 
$h_{\theta}(x) = e^{i \theta} h(x)$ and leaving all other $h_{\theta}(y)=h(y)$ for $y \neq x$. 
This produces a closed loop which rotating the wave $h$ along a circle. 

\paragraph{}
A bit surprisingly, even so at the end of the turn, while the matrix $L(h_{2\pi},G)=L(h_0,G)$ 
are the same, the eigenvalues of $L$ were in general permuted. 
The permutation group generated like this can be non-Abelian. 
It is not explored much yet. We have no ideas which groups can appear. 
 In the illustration section we only see example but we have code which allows
to experiment. If we want to associate eigenvectors of eigenvalues of $L$ 
with ``particles" generated by the wave $h$, the cycle structure of 
the group $\Pi(G,h)$ {\bf groups of eigenvalues and so the particles} naturally.

\paragraph{}
The interest in division algebras is first of all warranted by {\bf Hurwitz's theorem} 
which states that $\mathbb{R},\mathbb{C},\mathbb{H}$ or $\mathbb{O}$ is the 
complete list of {\bf normed real division algebras}. 
Nature appears to have a keen interest in these algebras: 
{\bf wave functions} in quantum mechanics are $\mathbb{C}$-valued, 
or $\mathbb{H}$-valued if one looks at {\bf spinors}. 
The units in $\mathbb{C}$ are the gauge group of electromagnetism, 
the units in $\mathbb{H}$ to the gauge group $SU(2)$ of the 
electro-weak interaction. The group 
$SU(3)$ naturally acts on $\mathbb{H}$ and also appears prominently in the construction of {\bf positive
curvature manifolds}, a category which is closely related to division algebras: all known examples in
that class are either spheres, or projective spaces over division algebras, or then one of four exceptions 
which are all $SU(3)$ based flag manifolds. The suggestion to link octonions $\mathbb{O}$ with gravity 
has appeared in various places (see the references in \cite{PosCurvBosons}). 

\paragraph{}
Each of the extensions $\mathbb{R} \to \mathbb{C} \to \mathbb{H} \to \mathbb{O}$ are algebraically and 
topologically essential: from $\mathbb{R}$ to $\mathbb{C}$, we make the unit sphere connected, 
from $\mathbb{C}$ to $\mathbb{H}$ we make the unit sphere simply connected, from 
$\mathbb{H}$ to $\mathbb{O}$, we also trivialize the second cohomology which manifests that there is
no Lie group structure any more on $\mathbb{S}^7$. 
Algebraically, we gain from $\mathbb{R} \to \mathbb{C}$ completeness, from 
$\mathbb{C} \to \mathbb{H}$ we lose not only commutativity, but the fundamental theorem of algebra.
From $\mathbb{C} \to mathbb{O}$ we also lose associativity and higher dimensional projective spaces.
As we will see already the transition from $\mathbb{R} \to \mathbb{C}$ is interesting for energized
complexes $(G,h)$ because the action of $\mathbb{U}$ on individual energy values produces a non-trivial
action on spectrum. 

\paragraph{}
We expect the spectral phase phenomenon to disappear again when looking at 
$\mathbb{K}=\mathbb{H}$ because the unit sphere in quaternions is simply connected. 
The assignment of eigenvalues is however is more problematic already in the 
quaternions case because there is not a strong fundamental 
theorem of algebra. The equation $x^2=-1$ for example has lots of solutions.

\paragraph{}
Division algebras are also closely linked to {\bf spinors}. The {\bf spin groups} ${\rm Spin}(p,1)$
are double covers of the indefinite Lie group $SO(p,1)$. They appear in Lorentzian space time: 
${\rm Spin}(2,1)=SL(2,\mathbb{R})$ in 2+1 space-time 
dimensions, ${\rm Spin}(3,1) = SL(2,\mathbb{C})$ in $3+1$ space-time dimension,
${\rm Spin}(5,1)=SL(2,\mathbb{H})$ in $5+1$ space-time dimension and 
${\rm Spin}(9,1)=SL(2,\mathbb{O})$ in $9+1$ space-time dimensions \cite{NcatlabSpinors}.

\paragraph{}
Spin groups are traditionally built through {\bf Clifford algebras}, a construct which
generalizes exterior algebras. In combinatorics, functions $h: G \to \mathbb{R}$ appear naturally as 
part of the discrete exterior algebra. There is no need to use a Clifford algebra to build it.
A $\mathbb{K}$-valued functions $h$ on a simplicial complex is in some sense a {\bf spinor-valued spinors}. 
Functions on sets of dimension $k$ are the $k$-forms, the 
{\bf exterior derivative} $d$ defines the {\bf Dirac operator} $D=d+d^*$ producing the block diagonal 
Hodge Laplacian $H=D^2 = \oplus H_k$ with ${\rm dim}({\rm ker}(H_k))=b_k$ 
is the $k$'th Betti number of $G$.  Many features appear to generalize however
\cite{KnillILAS,AmazingWorld,knillmckeansinger,DiscreteAtiyahSingerBott}.

\paragraph{}
The exterior calculus $G$ leads to the Barycentric refinement 
$\Gamma=(V,E)$ of $G$, where $V$ is the set of sets 
in $G$ and $E$ the set of pairs where one is contained in the other. This graph $\Gamma$ 
has a natural simplicial complex structure which is also called 
the {\bf Barycentric refinement} of $G$. One can then iterate the process of taking Barycentric 
refinements \cite{KnillBarycentric,KnillBarycentric2} to get universal fixed points 
which only depend on the maximal dimension of $G$.
We could extend the Barycentric refinement map $G \to G_1$ to energized 
complexes $(G,h) \to (G_1,h_1)$, where $h_1(x_1) = C \sum_{y \in x_1} h(y)$ is 
scaled so that ${\rm det}(L)=1$ and hope to get a fixed point. 
An other approach offers itself if we can assign for every  $G$ there exists 
unique solution $h:G \to \mathbb{R}$ of $\Psi(h) = \lambda$ on ${\rm det}(L)=1$. 
We would also like this get {\bf universal energized limit} $(G,h)$

\paragraph{}
Division algebras are also of interest in number theory and differential geometry. 
The primes in division algebra appear to have some combinatorial relations 
with Lepton and Hadron structures \cite{ParticlesPrimes}. 
 Also when studying the known even-dimensional positive curvature manifold 
types $\mathbb{RP}^{2d},\mathbb{S}^{2d}, \mathbb{HP}^d,\mathbb{OP}^2,
\mathbb{W}^6,\mathbb{E}^6,\mathbb{W}^{12},\mathbb{W}^{24}$, they have a curious
relation with force carriers in physics, where four exceptional positive 
curvature manifolds have a more complex cohomology, linking such manifolds 
with {\bf force carriers} having positive mass \cite{PosCurvBosons}. 
While these could well be just structural coincidences, it could also be a 
hint, that nature likes division algebras for building structure.

\paragraph{}
We saw here that the now order-dependent multiplicative energy $\prod_{x \in G} |h(x)|$ has
a meaning as a Study determinant of $L$ and that $H(G)=\sum_{x \in G} h(x)$ 
is the sum of the matrix entries of $g$.  
In this non-commutative setting like for quaternion-valued $h$, we had to recheck the proofs,
once an order on $G$ is given. This exercise also helped a bit to clarify more the proof of
unimodularity result originally found in 2016. Or the result that
if the energies take values in the group of units $SO(2)=U(1)$ of $\mathcal{C}$, or the 
units $SU(2)$ of $\mathcal{H}$, then $\overline{g}$ is the inverse of $L$. 
Especially, if the energy function $h$ takes values in the units in the 
ring $\mathbb{K}$ of integers of a cyclotomic field $\mathbb{K}$, then $L$ and its inverse are defined 
over $\mathbb{R}$. One can also imagine to use number fields within $\mathbb{C}$ 
like {\bf Gaussian integers} or {\bf Eisenstein integers}. The matrices $L,g$ take then values
in these integers at least if if the product of the $h$ values is a unit. 

\paragraph{}
The results extend algebraically to the {\bf strong ring} $\mathcal{G}$
generated by energized simplicial complexes, where $+$ is the disjoint
union of complexes and $\times$ the Cartesian product. We have
$\chi(G + H) = \chi(G) + \chi(H)$ and
$\chi(G \times H) = \chi(G) \chi(H)$ and
$L(G \times H) = L(G) \otimes L(H)$ and
${\rm det}(L(G \times H)) = {\rm det}(L(G)) {\rm det}(L(H))$.
We already have a shade of non-commutativity there because the connection
Laplacians are for the product the tensor products of the Laplacians. So,
the algebra in the strong ring naturally also calls for $\mathbb{K}$ to be
an algebra and we do not have to insist on commutativity of $\mathbb{K}$. 

\section{Questions}

\paragraph{}
For $\mathbb{K}=\mathbb{C}$, given an order on $G$, we have a natural {\bf spectral map}
$$ \Lambda: (h(x_1),\dots, h(x_n)) \mapsto (\lambda(x_1),\dots, \lambda(x_n)  \; . $$
The eigenvalue match of the last entry can be obtained by turning on the energy $h(x_n)$ and
follow track eigenvalues during the deformation.
Where is it invertible? Already for  $G=K_2=\{ \{1\},\{2\},\{1,2\} \}$
with energy function $h(\{1\})=a,h(\{2\})=b, h(\{1,2\})=c$ that the
characteristic polynomial is $p_{L}(q) = a b c-3 a b q-a c x+2 a q^2-b c q+2 b q^2+c q^2-q^3$,
that the Jacobian determinant of $\Lambda$ at $(a,b,c)=(1,1+\epsilon,1)$ is 
$$ \frac{\epsilon (4\epsilon +3)}{\sqrt{48+\epsilon (\epsilon (32 \epsilon (\epsilon+3)+157)+136)}} \; . $$
meaning that the determinant is zero in the counting case. The map $\Lambda$ is already
a very complicated map from $\mathbb{C}^3 \to \mathbb{C}^3$ 
but it is computable as we have formulas for the roots of cubic equations.
(A) Can one hear the energy $h$ of $G$ from the spectrum $\lambda$ of $L$? 
(B) If not, can one isospectrally deform the energy $h$? 
(C) What are the non-zero energy equilibria, fixed points $\Psi(h)=h$? 

\paragraph{}
Still in the case $\mathbb{K}=\mathbb{C}$, we can look at the
parametrization map $r: h=(h_1, \dots, h_n) \to M(n,\mathbb{C}$. 
Its rank is the rank of the $n \times n$ matrix $g=dr^T dr$ which is 
the first fundamental form rsp. the K\"ahler form in this setting. 
Because the map is $r$ is linear, the rank of this map is constant
and also the determinant of $g$. We have seen that the determinant 
is for positive dimensional simplicial complexes always divisible 
by $3$. We have in the K\"ahler section of this article give some examples.
What values can the determinant take? Why does the factor $3$ always appear? 

\paragraph{}
In all of our experiments also with larger complexes for real energies, the eigenvalues
depend in a monotone way on each energy value suggesting that in general, $d\Lambda$ 
(defined two paragraphs above) is 
a positive matrix with generically non-zero determinant. This makes sense physically because if we
pump in energy in a simplex, this should also quantum mechanically lead to larger energy
values for the Hamiltonian $L$. We are far from being able to prove this even in this real
case. Here is a more precise statement. For any fixed simplicial complex $G$ with $n$ sets,
is it true that the map $h \in (\mathbb{R}^*)^n \to \lambda_x(L_{G,h})$ assigning to the energized
complex $(G,h)$ the eigenvalue belonging to $x$ monotone. 
We actually believe that it is strictly monotone, even so
we have seen that the Jacobian matrix $d\Lambda$ can have zero determinant at some places.

\paragraph{}
We can also look at the map $\Psi: \mathbb{C}^n \to \mathbb{C}^n$ given by 
$$  (h(x_1),\dots, h(x_n)) \mapsto (g(x_1,x_1),\dots,g(x_n,x_n)) $$
assigning to the energies the diagonal Green function entries. These entries
are always of interest in physics (at least in a regularized way as most of the
time, the Green functions $g(x,x)$ do not exist in classical physics). 
What are the properties of this map? It maps the energies assigned
to the simplices to the self-energy. In the topological case, it is 
$\omega(x) \to \omega(x) (1-\chi(S(x))$.
The map $\Psi$ maps the diagonal entries of $L$ to the diagonal entries
of $g$.  These diagonal entries are 
$L(x,x) = \chi(W^-(x))$ and $g(x,x) = \chi(W^+(x))$. 

\paragraph{}
Are division algebra valued waves on simplicial complexes of physical consequences?
Let us call an energized simplicial complex $(G,h)$
to be {\bf Sarumpaet}, if the inserted energy to a simplex agrees with the
spectral energy. In other words, if there is a fixed point $\psi(h) =h$,
then $h$ serves as an energy equilibrium. Of course, there is always the
{\bf vacuum} $h=0$, where all matrices and eigenvalues are zero. But
we are interested in more interesting equilibria, similarly as the Einstein
equations give more interesting manifolds besides the trivial flat case.  
\footnote{This is motivated by the novel \cite{Egan}.}

\paragraph{}
If the fixed point equation can not be satisfied, one can then look at 
the minima $S(G) = {\rm min}_{h} |h-\lambda|_{2}$. The minimum exists
if $h$ takes values in the units $\mathbb{U}$ of $\mathbb{K}$. 
One can then further look at minima of $S(G)$ on the set of all simplicial 
complexes with $n$ elements and then see what happens with these minima 
if $n \to \infty$. This is still science fiction, but one can imagine 
the set of extrema to lead to quantities which do not 
depend on any input as both the geometry $G$ and the fields $h$ are given.

\paragraph{}
The construction of non-compact K\"ahler manifolds with given 
fundamental group is not a problem \cite{ABCKT}. The compact case
is difficult. One can ask whether it is possible to glue the manifolds
obtained here to compact manifolds. One could try to glue two such 
manifolds together for example. Attaching compact K\"ahler manifolds
to a complex $G$ with field $h$ would be more exciting. It would
also be interesting to understand the boundaries of the different
connectivity components. These are the places, where two or more
eigenvalues collide. 

\section{Code}

\paragraph{}
The following Mathematica code computes the permutation group for a given simplical
complex $G$ and energy vector $e$. As usual, the code can conveniently grabbed from
the ArXiv version of this paper. In this example, we take a linear
complex with 10 elements and as energy values the 10th root
of unities. The group has order 72. 

\begin{tiny}
\lstset{language=Mathematica} \lstset{frameround=fttt}
\begin{lstlisting}[frame=single]
Generate[A_]:=Delete[Union[Sort[Flatten[Map[Subsets,A],1]]],1];
G = Generate[{{1,2},{2,3},{3,4},{5,6}}]; n=Length[G];  Q=SubsetQ;
S=Table[-(-1)^Length[G[[k]]]*If[k ==l,1,0],{k,n},{l,n}]; 
star[x_]:=Module[{u={}},Do[v=G[[k]];If[Q[v,x],u=Append[u,v]],{k,n}];u];
core[x_]:=Module[{u={}},Do[v=G[[k]];If[Q[x,v],u=Append[u,v]],{k,n}];u];
Wminus    = Table[Intersection[core[G[[k]]],core[G[[l]]]],{k,n},{l,n}];
Wplus     = Table[Intersection[star[G[[k]]],star[G[[l]]]],{k,n},{l,n}];
e=Table[Exp[1.0*2Pi I k/n],{k,n}]; 
EN[A_]:=If[A=={},0,Sum[e[[Position[G,A[[k]]][[1,1]]]],{k,Length[A]}]];
L=  Table[EN[Wminus[[k,l]]],   {k,n},{l,n}];   V=Eigenvalues[1.0*L];
g=S.Table[EN[Wplus[[k,l]]],    {k,n},{l,n}].S; Chop[Conjugate[g].L];
TrackEigenvalue[m_,w_]:=Module[{t=0,XX=V[[m]],q=1},
 Do[ e1=e;  e1[[w]]=e[[w]]*Exp[I t];
 EN[A_]:=If[A=={},0,Sum[e1[[Position[G,A[[k]]][[1,1]]]],{k,Length[A]}]];
 L1=Table[EN[Wminus[[k,l]]], {k,n},{l,n}]; V1=Eigenvalues[1.0*L1]; 
 min=n; Do[If[Abs[XX-V1[[k]]]<min,min=Abs[XX-V1[[k]]];q=k],{k,n}];
 XX=V1[[q]], {t,0,2Pi,2Pi/500}];
 min=n;Do[If[Abs[XX-V[[k]]]<min,min=Abs[XX-V[[k]]];q=k],{k,n}];q]; 
Perm[w_]:=Table[TrackEigenvalue[m,w],{m,n}]; 
GroupOrder[PermutationGroup[Table[PermutationCycles[Perm[w]],{w,n}]]]
\end{lstlisting}
\end{tiny}

\paragraph{}
The following lines compute the rank of the matrix $dr^T dr$ in the
case, where $G=K_4$. Just by changing $G$ one can compute the rank
for any set of sets. We also make the prime factorization because
we observe that for some strange reason, the determinant of $dr^T dr$
is always divisible by $3$ in positive dimensional cases.
The case of a sets of sets ({\bf multi-graph}) example like $G=\{ \{1,2,3\} \}$ 
are also considered zero dimensional, even so the element is $2$ dimensional.
The connection matrix is the $1 \times 1$ matrix $1$. 

\begin{tiny}
\lstset{language=Mathematica} \lstset{frameround=fttt}
\begin{lstlisting}[frame=single]
Generate[A_]:=Delete[Union[Sort[Flatten[Map[Subsets,A],1]]],1];
G1=Generate[{{1,2,3,4,5},{3,4,5,6,7}}];  G2=Generate[{Range[4]}];
G=G2; n=Length[G]; SQ=SubsetQ; 
c[x_]:=Module[{u={}},Do[v=G[[k]];If[SQ[x,v],u=Append[u,v]],{k,n}];u];
Wminus = Table[Intersection[c[G[[k]]],c[G[[l]]]],{k,n},{l,n}];
Q[e_]:=Module[{en,L},
  en[A_]:=If[A=={},0,Sum[e[[Position[G,A[[k]]][[1,1]]]],{k,Length[A]}]];
  L = Table[en[Wminus[[k,l]]],  {k,n},{l,n}]; Flatten[L]];
Id=IdentityMatrix[n]; A=Transpose[Table[Q[Id[[k]]],{k,n}]];         
Kaehler=Transpose[A].A;
Print[{Det[Kaehler],"Factors:",FactorInteger[Det[Kaehler]]}];
\end{lstlisting}
\end{tiny}

\begin{figure}[!htpb]
\scalebox{0.8}{\includegraphics{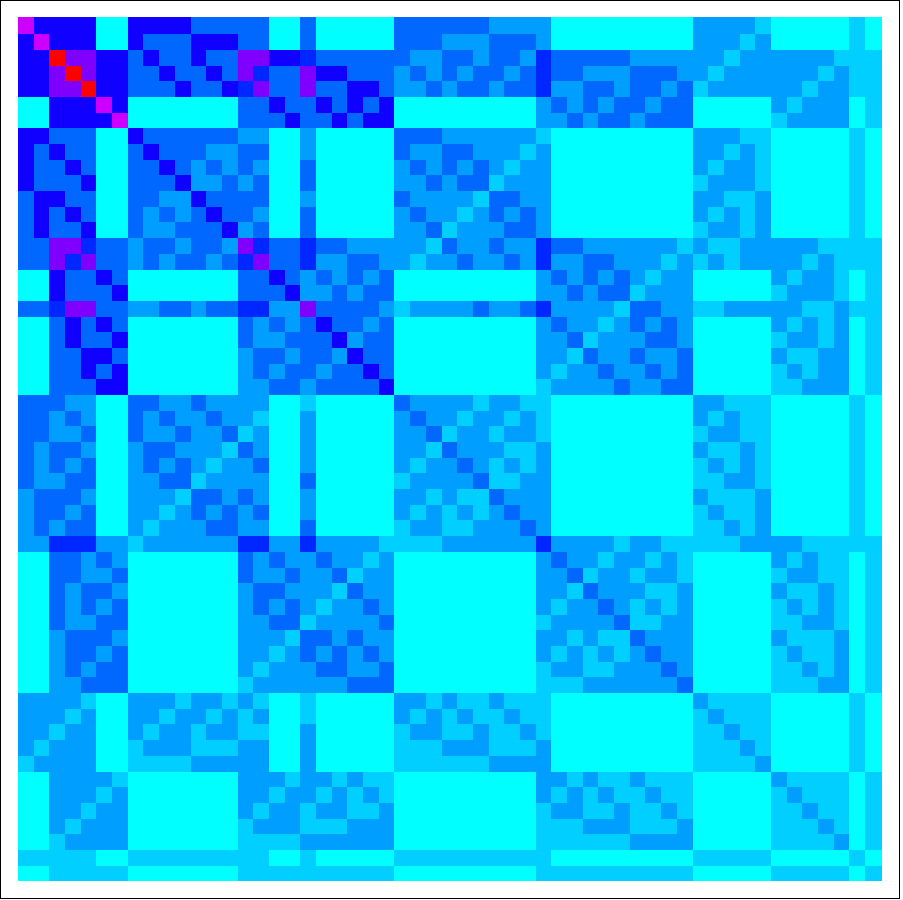}}
\label{kaehler}
\caption{
The K\"aehler bilinear form $A=dr^T dt$ in the case of $G=G1$ of the above code
which is the simplicial complex generated by the two sets 
$\{1,2,3,4,5\},\{3,4,5,6,7\}$ which has $55$ elements. 
The determinant of $A$ is $3^{113} 5^{7} 7^7$ in this case. 
}
\end{figure}

\begin{figure}[!htpb]
\scalebox{0.8}{\includegraphics{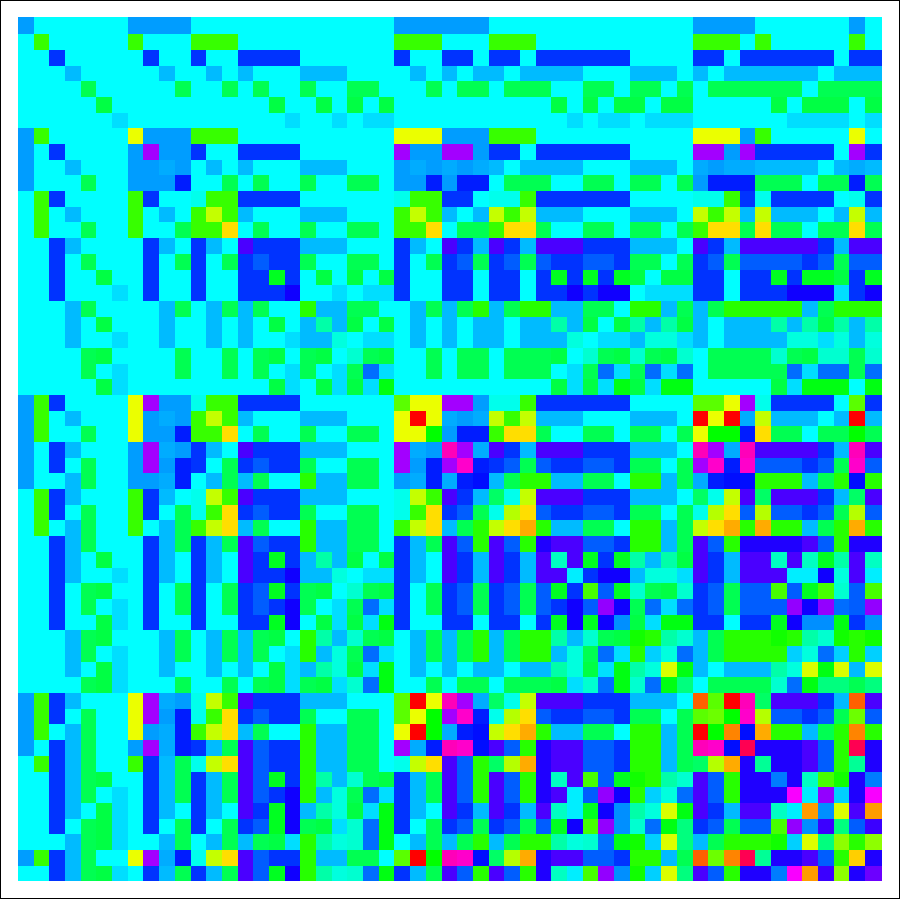}}
\label{matrix}
\caption{
We see the $55 \times 55$ matrix $L$ in the case of the complex generated by 
$\{ \{1,2,3,4,5\},\{3,4,5,6,7\} \}$. The energy was chosen to be random
with elements in $[0,1]$. 
}
\end{figure}

\section*{Appendix: non-commutative determinants}

\paragraph{}
There are various determinant constructions for matrices taking values in 
non-commutative rings. First of all, there is the {\bf Leibniz determinant}
${\rm det}(L) = \sum_{\sigma} {\rm sign}(\sigma) L_{1,\sigma(1)} \cdots L_{n \sigma(n)}$
which does not satisfy the {\bf Cauchy-Binet property}
${\rm det}(A B) = {\rm det}(A) {\rm det}(B)$.
Already for $1 \times 1$ matrices $A=[[a]], B=[[b]]$ one has ${\rm det}(AB)=ab$ and
${\rm det}(BA)=ba$. We also can consider the
{\bf Study determinant} \cite{Study1920} as well as the {\bf Dieudonn\'e determinant}
\cite{Dieudonne1943} which both do satisfy the Cauchy-Binet relation. Both have their
uses, advantages and disadvantages. 

\paragraph{}
The Leibniz determinant is defined for any matrix
over a commutative or non-commutative ring $\mathbb{K}$ and does not even use associativity
if we make a convention about where multiplication starts (like from the right if no brackets are 
used). The ring $\mathbb{K}$ can also be a division algebra like 
$\mathbb{O}$ which is no more a ring as the multiplication is no more associative. It could be a 
{\bf normed Jordan algebra} or {\bf semi-simple Lie algebra}. Maybe of more physics interest is when
$\mathbb{K}$ is $C^*$-algebra of bounded linear operators on a Hilbert space or a Banach algebra.
Of special interests are von Neumann algebras. If $\mathbb{K}$ is a quadratic field one is in 
an algebraic number field setting, which could be of interest in number theory. 

\paragraph{}
The {\bf Dieudonn\'e determinant} is covered in 
\cite{Dieudonne1943,ArtinGeometricAlgebra,RosenbergKTheory}.
We especially can follow Artin (Chapter IV) or Rosenberg (Section 2.2).
If $\mathbb{K}$ is a division algebra denote by $\overline{\overline{K}}$ its Abelianization.
It is given in terms of the multiplicative group $\mathbb{K}^*=\mathbb{K} \setminus \{0\}$ 
of $\mathbb{K}$ as $\overline{\mathbb{K}} = \mathbb{K}^*/[\mathbb{K}^*,\mathbb{K}^*] \cup \{0\}$. 
While in $\mathbb{K}=\mathbb{C}$ we have $\overline{x}=x$ in the quaternions, this is no more
the case for quaternions. We have $\overline{-1}=1$ for example. So, even if $-1$ is a real
number within the quaternions and we are working in $\mathbb{K}=\mathbb{H}$ then $\overline{-1}
=1$. The reason is that the commutator $[i,j]=i j i^{-1} j^{-1} = k (-i) (-j) = k^2=-1$
contains $-1$. 

\paragraph{}
The Dieudonn\'e determinant of a matrix $A$ is axiomatically defined as a function $M(n,\mathbb{K})
\to \overline{\mathbb{K}}$ which has the following properties:
(i) ${\rm det}(1)=1$, (ii) Multiplying a row of $A$ with $\mu$ from the left 
multiplies the determinant by $\overline{\mu}$ .
(iii) Adding a row of $A$ to an other does not change the determinant.
It follows: (iv) $A$ is singular if and only if ${\rm det}(A)=0$.
(v) If two rows are interchanged, then ${\rm det}(A)$ is multiplied with $\overline{(-1)}$. 
(vi) Multiplying $A$ from the right by $\mu$ multiplies the determinant by $\overline{\mu}$. 
For the following, also see \cite{Brenner1968}:
(vii) ${\rm det}(A^T) = {\rm det}(A)$. 
(viii) The Laplace expansion works with respect to any row or column \cite{Brenner1968}.
(ix) There is a Leibniz formula 
${\rm det}(A) = \sum_{\sigma} (-1)^{\sigma} A_{1 \sigma(1)} \cdots A_{n \sigma(n)} w_{\sigma}$
holds, where each $w_{\sigma}$ is a commutator of the multiplicative group. 

\paragraph{} The explicit computation that the commutator $c = a b a^{-1} b^{-1}$ can 
be realized through row reduction is shown as Theorem~(4.2) in \cite{ArtinGeometricAlgebra}:
$\left[ \begin{array}{cc} 1 & 0 \\ 
                          0 & 1 \end{array} \right] \to$ 
$\left[ \begin{array}{cc} 1 & 0 \\ 
                          a^{-1} & 1 \end{array} \right] \to$ 
$\left[ \begin{array}{cc} 0 & -a \\ 
                          a^{-1} & 1 \end{array} \right] \to$ 
$\left[ \begin{array}{cc} 0 & -a \\ 
                          a^{-1} & b^{-1} \end{array} \right] \to$ 
$\left[ \begin{array}{cc} a b a^{-1} & 0 \\ 
                          a^{-1} & b^{-1} \end{array} \right] \to$ 
$\left[ \begin{array}{cc} a b a^{-1} & 0 \\ 
                          1          & b^{-1} \end{array} \right] \to$ 
$\left[ \begin{array}{cc} 0          & -c \\ 
                          1          & b^{-1} \end{array} \right] \to$ 
$\left[ \begin{array}{cc} 0          & -c \\ 
                          1          & c   \end{array} \right] \to$ 
$\left[ \begin{array}{cc} 1          &  0 \\ 
                          1          & c   \end{array} \right] \to$ 
$\left[ \begin{array}{cc} 1          &  0 \\ 
                          0          & c   \end{array} \right]$. This sequence
of steps can be performed in general in the last two rows of an $n \times n$ matrix. 

\paragraph{}
The row reduction step shows that $A$ can be written as $A=B D(1, \dots, \mu)$ with 
an unimodular $B$ and some $\mu \in \mathbb{K}$. 
In that case ${\rm det}(A) = \overline{\mu}$. If $\overline{\mu}=1$, then $\mu$ is 
a product of commutators. This means $SL(1,\mathbb{K}) = \overline{\mathbb{K}}$.
Artin gives the following examples: 
$$ {\rm det}( \left[ \begin{array}{cc} 0 & b  \\ 
                                       c & d  \end{array} \right] ) = \overline{-cb} $$
For nonzero $a$ values, one has
$$ {\rm det}( \left[ \begin{array}{cc} a & b  \\ 
                                       c & d  \end{array} \right] ) = \overline{ad-a c a^{-1} b} $$
The example shows that one can not factor out $b$ in the second row: 
$$ {\rm det}( \left[ \begin{array}{cc} 1 & a  \\
                                       b & ab  \end{array} \right] ) = \overline{ab-ba} $$

\paragraph{}
The {\bf Study determinant}
$\det(M)$ is the unique multiplicative functional on $M(n,\mathbb{H})$ 
that is zero exactly on singular matrices and which has the property that
$\det(1+r E_{ij})=1$ for $i \neq j$ and $r \in \mathbb{H}$. It agrees with
$|\det(A)|$ on $\mathbb{C}$ and satisfies ${\rm det}(A) = |\prod_k \lambda_k|$. 
The Study determinant was introduced in \cite{Study1920} and is also part of the axiomatization
of \cite{Dieudonne1943}. Like the Dieudonn\'e determinant, it is defined by row 
reducing the matrix to an upper triangular $2 \times 2$ block matrix and then 
conjugating the $2 \times 2$ case. The Study determinant is just the norm of the Dieudonn\'e 
determinant and contains therefore in general a bit less information than the later. 
If $A$ is upper triangular, then ${\rm sdet}(A) =\prod_j |A_{jj}|$
as well as ${\rm sdet}(A) = \prod_j |\lambda_j|$ if $\lambda_j$ are the 
eigenvalues. For complex matrices the Study determinant satisfies 
${\rm sdet}(A) = |\det(A)|$. 

\paragraph{}
The eigenvalues of $A$ in a division algebra can be defined 
by noting that $A$ is similar to an upper
triangular matrix by Gaussian elimination. One can also conjugate $A$ 
over the quaternions to a complex Jordan normal form matrix $J$.
If $B$ is the complexification of $A$, meaning 
$B= \left[ \begin{array}{cc} A_1 &  - A_2^* \\
                             A_2 &    A_1^* \end{array} \right]$ if
$A=A_1 + j A_2$ where $A_i$ are complex, then the spectrum of $B$
is $\lambda_1,\lambda_1^*  \dots \lambda_n, \lambda_n^*$. 
They satisfy the eigenvalue equation $A u = u \lambda$. 

\paragraph{}

\bibliographystyle{plain}

\end{document}